\documentclass[conference]{IEEEtran}
\usepackage{url}
\usepackage{graphicx}
\usepackage{algorithm}
\usepackage[noend]{algpseudocode}
\usepackage{subfigure}
\usepackage{amssymb, amsmath, graphicx, charter, latexsym}
\usepackage{amsthm}
\usepackage{setspace}
\usepackage{enumitem}
\usepackage{ragged2e}
\usepackage{times}
\usepackage{mathtools}
\usepackage{tabu}
\usepackage{epstopdf}
\usepackage{multirow}
\usepackage{verbatim}
\usepackage{color}
\usepackage{lipsum}
\usepackage{bm}

\usepackage{mdframed}
\usepackage{bm}
\usepackage[cal=cm]{mathalfa}
\usepackage{lipsum}
\usepackage[compress,sort]{cite}

\let\emptyset\varnothing

\DeclareMathAlphabet{\mathpzc}{OT1}{pzc}{m}{it}
\DeclarePairedDelimiterX{\bkt}[1]{(}{)}{ #1}
\DeclarePairedDelimiterX{\sbkt}[1]{[}{]}{ #1}
\DeclarePairedDelimiterX{\lbkt}[1]{\{}{\}}{ #1}

\theoremstyle{definition}

\newcolumntype{x}[1]{>{\raggedright\arraybackslash}p{#1}}
\def\baselinestretch{1}

\newmdtheoremenv[linewidth=0.95pt, skipabove=6pt, skipbelow=6pt]{lemma}{Lemma}
\newmdtheoremenv[linewidth=0.95pt, skipabove=6pt, skipbelow=6pt]{theorem}{Theorem}
\newmdtheoremenv[linewidth=0.95pt, skipabove=6pt, skipbelow=6pt]{protocol}{Protocol}
\newmdtheoremenv[linewidth=0.95pt, skipabove=6pt, skipbelow=6pt]{motion_protocol}{Motion Protocol}
\newmdtheoremenv[linewidth=0.95pt, skipabove=6pt, skipbelow=6pt]{policy}{Policy}
\newmdtheoremenv[linewidth=0.95pt, skipabove=6pt, skipbelow=6pt]{constraint}{Constraint}
\newmdtheoremenv[linewidth=0.95pt, skipabove=6pt, skipbelow=6pt]{rules}{Rules}

\begin{document}

\title{Safe Intersection Management for Mixed Transportation Systems with Human-Driven and Autonomous Vehicles}
\author{
	\IEEEauthorblockN{Xi Liu}
	\IEEEauthorblockA{\small 
		Texas A\&M University\\
		College Station, TX 77843, USA\\
		xiliu@tamu.edu
	}\and
	\IEEEauthorblockN{Ping-Chun Hsieh}
	\IEEEauthorblockA{\small 
		Texas A\&M University\\
		College Station, TX 77843, USA\\
		pingchun.hsieh@tamu.edu
	}\and
	\IEEEauthorblockN{P. R. Kumar}
	\IEEEauthorblockA{\small 
		Texas A\&M University\\
		College Station, TX 77843, USA\\
		prk.tamu@gmail.com
	}
}

\maketitle
\begin{singlespacing}
\begin{abstract}
Most recent studies on establishing intersection safety focus on the situation where all vehicles are fully autonomous. However, currently most vehicles are human-driven and so we will need to transition through regimes featuring a varying proportion of human-driven vehicles ranging from 100\% to 0\% before realizing such a fully autonomous future -- if ever. We will therefore need to address the safety of hybrid systems featuring an arbitrary mixture of human-driven and autonomous vehicles. In fact recent incidents involving autonomous vehicles have already highlighted the need to study the safety of autonomous vehicles co-existing with human-driven vehicles. Motivated by this we address the design of provably safe intersection management for mixed traffic consisting of a mix of human-driven vehicles (HVs) as well as autonomous vehicles (AVs).

To analyze such mixed traffic, we model HVs as nearsighted and with relatively loose constraints, permitting worst-case behavior while AVs are considered as capable of following much tighter constraints. HVs are allowed freedom to change their speed at any time while AVs are only allowed to change their speed at the beginning of a time slot through a Model Predictive Controller (MPC). AVs are assumed to possess a shorter response time and stronger braking capability than HVs in collision avoidance. Moreover, AVs obtain the permissions of passing through the intersection through vehicle-to-infrastructure (V2I) communication, while HVs achieve the same objective by following traffic lights.

Taking the above differences into consideration, we propose a provably safe intersection management for mixed traffic comprised of an MPC-based protocol for AVs, a permission assignment policy for AVs along with a coordination protocol for traffic lights. In order to respect the distinctiveness of HVs, the proposed protocol ensures that the traffic lights as well as the semantic meanings of their colors are consistent with current practice. A formal proof of safety of the system under the proposed combined protocol is provided. 

\end{abstract}
\end{singlespacing}
\section{Introduction}
\label{section:introduction}

According to the European Union community road accident database CARE, intersection related fatalities account for more than 20\% of the fatalities in the European Union during 2001-2010 \cite{c7_broughton2013traffic}. Similarly in the US, 40\% of the crashes and 21.5\% of the traffic fatalities are intersection related \cite{c8_nhtsa}. Autonomous vehicles (AVs) have been regarded as being able to mitigate this issue. Their technology has been advancing thanks to many related projects, e.g., the ITS program in the US, the EUREKA Prometheus Project in European Union, the ITS initiative program in Japan and the Self-Driving Car Project by Google. In the 2007 DARPA Urban Challenge, six out of 32 AVs completed the race, indicating feasibility of fully autonomous driving in urban environment \cite{c1_darpa2007}. Besides, supportive standards and laws have been adopted for AVs clearing their way to enter the market. To support communication, the IEEE Wireless Access for Vehicle Environment (WAVE) standard and IEEE 802.11p based Directed Short Range Communication (DSRC) standards for vehicle-to-vehicle (V2V) communication \cite{c2_ieee2013ieee} and vehicle-to-infrastructure (V2I) communications \cite{c3_dsrc2009dedicated} have been developed. Testing of AVs has been legally permitted on public roads in California, Michigan, Florida, Nevada, Arizona, North Dakota, Tennessee, and the District of Columbia \cite{c4_ncsl}.

For the reasons discussed above, intersection safety with autonomous vehicles has received growing research attention. However, most recent studies focus on establishing intersection safety in the context that \emph{all} vehicles passing through the intersection are fully autonomous, i.e., all vehicles have computation capability, are able to negotiate with one another, or receive instructions from a centralized controller. Provable safety is attained in \cite{c17_kowshik2011provable} through assigning AVs non-overlapping time slots for accessing the intersection to avoid collisions. This assignment is accomplished by assuming that all vehicles have GPS, V2I communication, in-vehicle sensing and computation. More strictly, in \cite{c18_carlino2013auction,bashiri2017platoon} only one AV or platooning is allowed to enter the intersection per time. The order to enter is obtained via V2I communication. Similarly, safety is improved in \cite{c13_kamal2015vehicle} by preventing pairs of conflicting AVs from approaching their cross-collision point at the same time. This prevention mechanism requires that all vehicles are equipped with V2I communication and a computation module. In \cite{qian2015decentralized,c14_makarem2013model}, a higher requirement is imposed on vehicles: all of them are able to access the system state through wireless communication while solving Model Predictive Control (MPC) based optimization problems. In \cite{c15_de2013autonomous}, intersection safety is attained by AVs sequentially solve optimization problems with the solution of previous AVs being passed to succeeding AVs as additional constraints through V2V communication.

The fraction of AVs in present traffic is negligibly small and it is predicted that by 2030, only 50\% of vehicles will be autonomous \cite{c9_de2014network}. Therefore there may be a long transition period, during which a mix of HVs as well as AVs share the same intersection. So far, however, there has been little rigorous attention to safe intersection management in mixed traffic (a review is given in Section \ref{sec:related_works}). Moreover, it is very difficult to extend studies on homogeneous traffic to mixed traffic as AVs and HVs are significantly different. First, HVs are operated by human drivers and thus ought to have the freedom to change speed at any time. In comparison, AVs are usually controlled by controllers that are only able to adjust control inputs at the beginning of each time slot. Second, human drivers should not be expected to follow complex commands, while AVs are considered as capable of following much tighter constraints. Third, with respect to braking to avoid collision, human drivers may be unable to response as quickly as AVs. Finally, AVs learn about each other's intentions through V2V communication and receive commands from intersections through V2I communication, while human drivers achieve the same goal by reading signaling lights on other vehicle, such as turn signals, and intersection traffic lights.

Motivated by the issues presented above, we address the design of provably safe intersection management for mixed traffic consisting of a mix of HVs as well as AVs. This paper continues our work on single-lane and multi-lane traffic \cite{c10_liu2015towards,c11_liu2016towards,hsieh2017throughput}, and is apparently the first to address provably safe intersection management for mixed traffic under above considerations. 

\noindent \textbf{Contributions} The contributions of our work can be summarized as follows.
\begin{itemize}
    \item A novel intersection management is proposed for mixed traffic consisting of a mix of HVs and AVs. The management maximally respect the habits of human drivers in current practice.
    \item HVs and AVs are represented by different models that capture their differences in control freedom, command complexity, braking response capability, and communication capability.
    \item The proposed management takes into account the difference between AVs and HVs.
    \item A formal proof for the safety of the proposed management is provided.
\end{itemize}

The rest of this paper is organized as follows. In Section \ref{sec:related_works}, we provide a brief summary of related work in intersection management for mixed traffic. In Section \ref{sec:problem_modeling}, we present our models of intersections, HVs and AVs. Intersection management for mixed traffic in combination with a coordination protocol for traffic lights is investigated in Section \ref{sec:methods}, and subsequently, a formal proof of safety is provided. Concluding remarks follow in Section \ref{sec:conclusion}.

\section{Related Works}\label{sec:related_works}

Very few studies have investigated intersection management for mixed traffic consisting both AVs and HVs. In \cite{c20_dresner2007sharing}, mixed traffic is proposed to be managed through redeploying existing infrastructure-traffic lights. As separate traffic light is assumed to be responsible for each lane, each lane is successively granted with a green light during a small portion of time while HVs in other lanes are stopped by red lights. A formal proof on system safety is missing. In \cite{bento2013intelligent}, an HV is viewed as a virtual AV. Every time an HV arrives, the intersection reserves an exclusive time for it, and is informed of its right-of-way by the human-dedicated traffic lights. This design faces difficulties since it imposes additional constraints on traffic lights and lane infrastructure. In \cite{c19_qian2014priority}, AVs are controlled to respect priority-preserving order provided by intersection, while for HVs, the intersection is responsible to create a virtual request, process it and inform them through traffic lights. However, when there are low-priority AVs in front, HVs are blocked out of the intersection in spite the green traffic lights leading to a mismatch with human driver expectations. The study also does not model the differences between AVs and HVs except for their communication capabilities. In \cite{altche2017algorithm,ahn2018safety} intersection management is investigated for ``semi-autonomous" vehicles, i.e., vehicles that support control-input override and V2I communication. Safety is guaranteed by allowing intersection to override control inputs from human drivers when they would result in an unsafe or blocked situation, which presents difficulties in current practice as most present HVs does not support such functions.

\section{Problem Formulation}\label{sec:problem_modeling}

In the context of mixed traffic, intersection modeling and vehicle modeling are expected to fully respect the differences between AVs and HVs, while avoiding unnecessary changes to current infrastructure. The differences between AVs and HVs can be summarized as below aspects.

\begin{enumerate}
    \item \emph{Control freedom}: HVs have the freedom to change a control input at any time, while AVs are only able to adjust control input at periodic discrete times.
    
    \item \emph{Command complexity}: HVs should not be expected to follow complex commands, whereas AVs are viewed as able to follow tighter constraints and complex protocols.

    \item \emph{Braking response capability}: HVs cannot brake as hard as AVs. HVs are unable to response as quickly as AVs. 
    
    \item \emph{Communication capability}: AVs can communicate with each through V2V communication and with the intersection bidirectionally by V2I communication, while HVs can only read signal lights on vehicles and intersections ``unidirectionally'', where 
``unidirectionally'' here means that HVs cannot directly send any acknowledgement or talk to intersections.
\end{enumerate}

\subsection{Intersection Modeling}

We consider a general intersection scenario as shown in Fig. \ref{fig:intersection}. We assume vehicles follow legal paths represented by solid curves. Paths are numbered, and intersect at several ``conflict points''. Around each conflict point, a small area accounting for the length of a vehicle is defined and represented by a solid circle referred to as a ``collision area''. One target of safe intersection management is to prevent two vehicles in one system from entering the same collision area simultaneously. Our intersection model allows for the case that a lane supports more than one traffic flow, e.g., ``straight'' and ``right-turn''. Consequently, a vehicle with permission may fail to enter the intersection if it is blocked by a vehicle ahead that is still waiting a permission. Although Fig. \ref{fig:intersection} shows three lanes for each in-road, the proposed scheme extends to more general scenarios, where the first $l_{1}$ lanes are ``left-turn only'', the next $l_{2}$ lanes are ``left-turn or straight'', the next $l_{3}$ lanes are ``straight only'', the next $l_{4}$ lanes are ``straight or right-turn'', and the last $l_{5}$ lanes are ``right only''. 

\vspace{-3mm}
\begin{figure}[thpb]
\centering
\includegraphics[width=.48\textwidth, height=.45\textwidth]{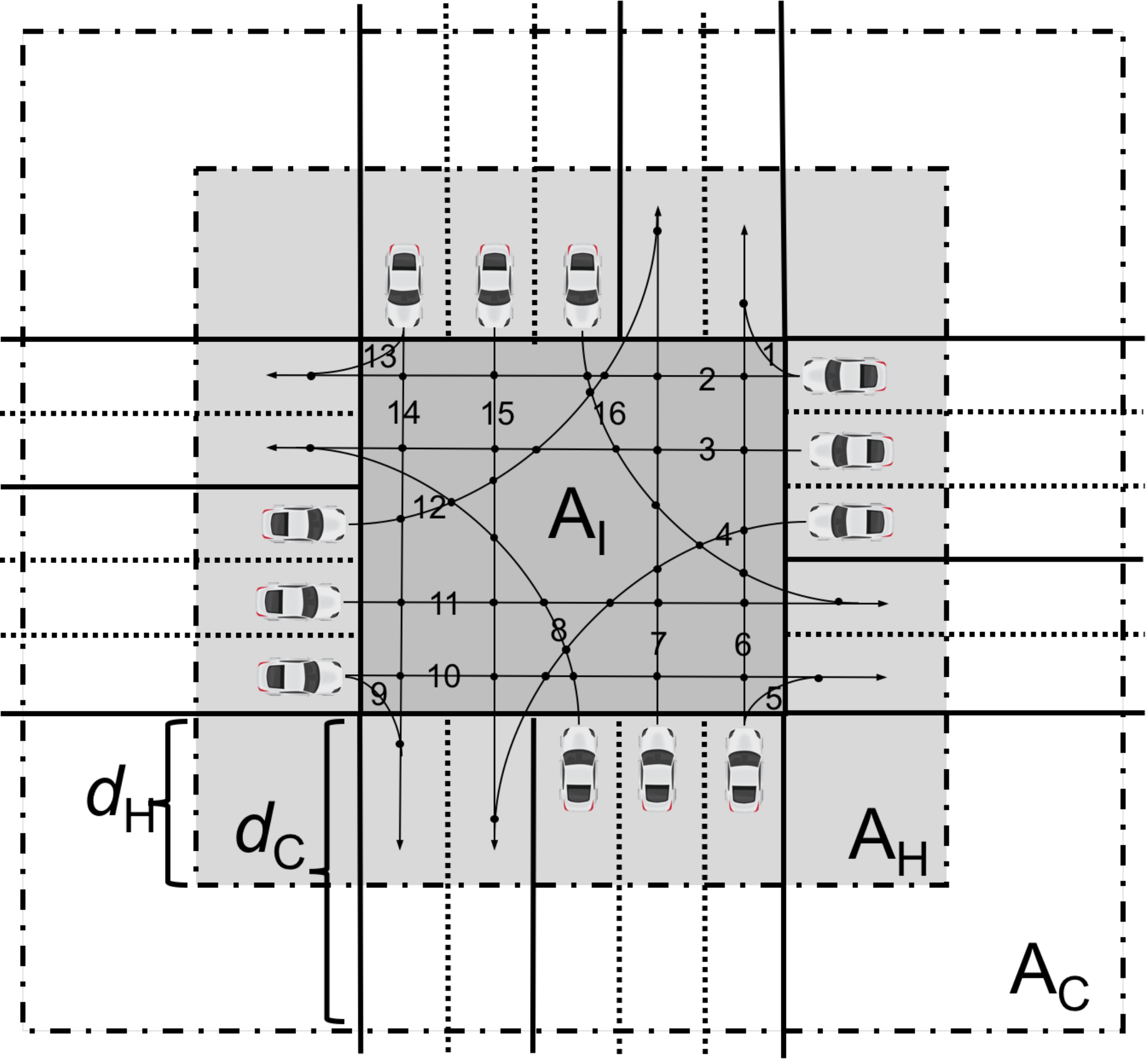}
\caption{General intersection scenario with crossing, merging, and non-conflicting paths}
\label{fig:intersection}
\end{figure}

In Fig. \ref{fig:intersection}, we denote the area inside the intersection by $A_{I}$, entering which will be referred to as ``entering the intersection''. Supposing the range of reliable V2I communication to be $d_{C}$, we denote by $A_{C}$ the area that is within a distance $d_{C}$ from $A_{I}$. AVs are expected to stay connected with the intersection inside area $A_{C}\cup A_{I}$. Unlike AVs, HVs plan their motion through the intersection through traffic lights. Motivated by this, we suppose that HVs within a distance $d_{H}$ notice the color of the traffic light before entering $A_{I}$. We denote by $A_{H}$ the area within a distance $d_{H}$ of $A_{I}$. Since $d_{H}$ is line-of-sight, we assume $d_{H}<d_{C}$ in Fig. \ref{fig:intersection}. Given these differences between HVs an AVs, our goal is to design an intersection management scheme that avoids collision among vehicles at any time and place in the intersection.

\subsection{Vehicle Modeling}
\begin{figure}[thpb]
\centering
\includegraphics[width=.5\textwidth, height=.41\textwidth]{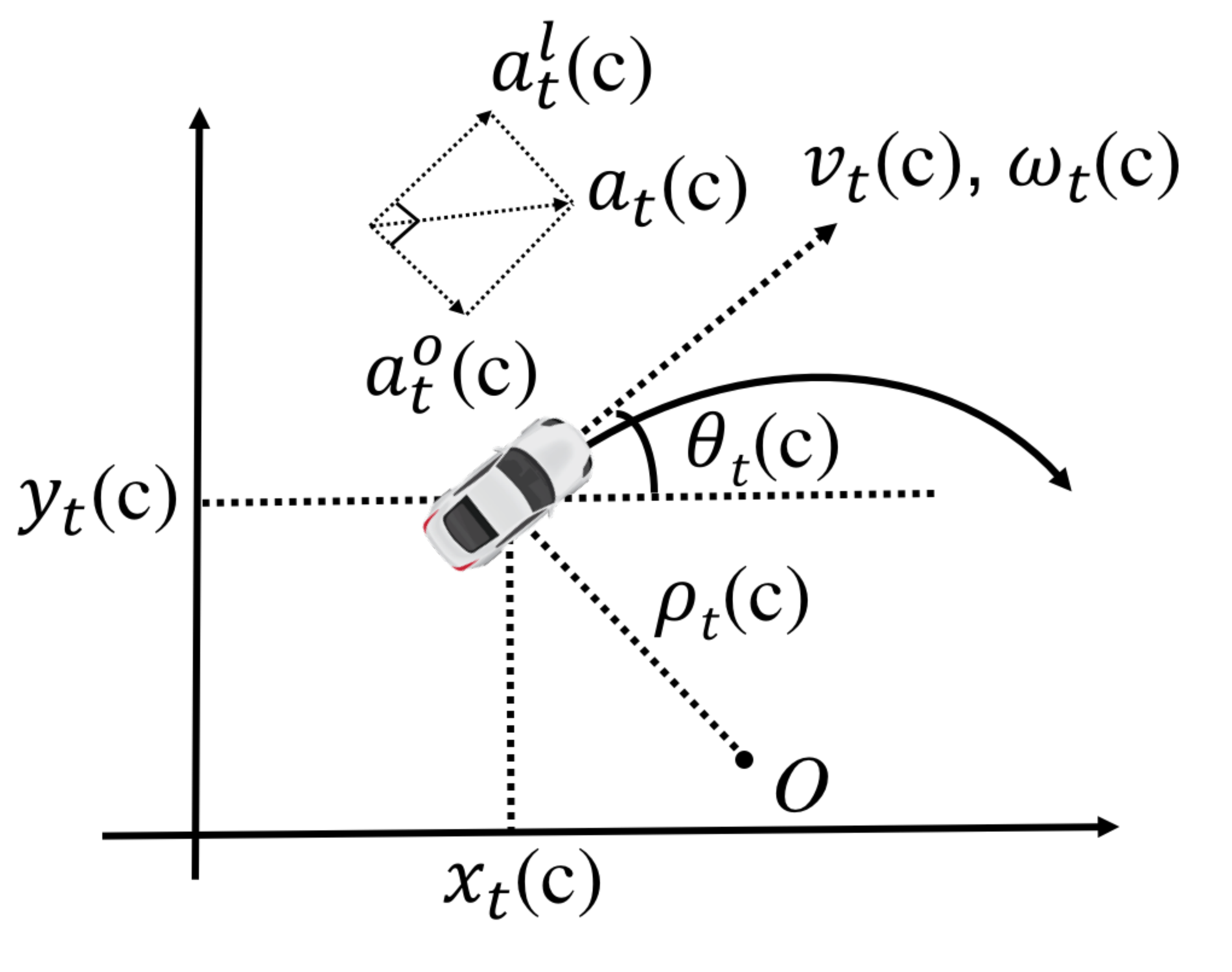}
\caption{Illustration of the unicycle kinematic model for vehicle $c$. The solid curve depicts the trajectory of the vehicle. }
\label{fig:Kinematic}
\end{figure}

We model the motion of vehicles by a unicycle kinematic model. As shown in Fig. \ref{fig:Kinematic}, at time $t$, the state of vehicle $c$ is $ \mathbf{x}_{t}(c):=(x_{t}(c), y_{t}(c), \theta_{t}(c))^{T}$, capturing its position and orientation. Longitudinal deceleration is along the forward direction and is denoted by $a_{t}^{\ell}(c)$. Centripetal deceleration points to the turning center $O$ and is denoted by $a_{t}^{o}$, resulting in an equivalent acceleration $a_{t}(c) = \sqrt{(a_{t}^{\ell}(c))^{2}+(a_{t}^{o}(c))^{2}}$. $\rho_{t}(c)$ denotes the predefined turning radius of $c$ at position $p_{t}(c)$, as defined by the predefined lane markers. The input vector is $\mathbf{u}_{t}:=(v_{t}(c),\omega_{t}(c))$ denoting velocity control and steering control. Where there is no scope for confusion we omit the vehicle identifier $c$. 

The physical constraints on the kinematic model are: (1) $v_{t}\in[0,v_{max}]$. (2) $\theta_{t}\in[\theta_{min},\theta_{max}]$. (3) $|\rho_{t}|\in[\rho_{min},+\infty)$. (4) $|\omega_{t}|\in[0, v_{max}/\rho_{min}]$; (5) $\varDelta{v_{t}}\in[a_{min}h,a_{max}h]$. Here $v_{max}\geq 0$ denotes the speed limit at the intersection, while $\theta_{min}<0$ and $\theta_{max}>0$ correspond to the minimum and maximum turning according to its angels. The quantity $a_{min}<0$ denotes the most rapid deceleration, where $a_{min}=-\sqrt{(a^{l}_{min})^{2}+(a^{o}_{min})^{2}}$, $a^{l}_{min}<0$ is the maximum achievable lateral deceleration (i.e., most rapid brake) and $a^{o}_{min}<0$ is maximum achievable centripetal deceleration by vehicle tires, $a_{max}>0$ is the maximum achievable acceleration, i.e., maximum throttle, and the minimum turning radius along trajectories.

Control signals for AVs are determined at the beginning of each time slot $[t,t+h)$ and maintained constant during it. In contract, control inputs for HVs are allowed to change at any time. The kinematic equation of an vehicle is
\begin{align}
\mathbf{x}_{t+h}:=f(\mathbf{x}_{t}, \mathbf{u}_{t}).
\end{align}
Then, for an AV when $\omega_{t}\neq0$, we express $f$ by
\begin{align*}
x_{t+h}:&=2\dfrac{v_{t}}{\omega_{t}}\sin(\dfrac{1}{2}\omega_{t}h)\cos(\theta_{t}+\dfrac{1}{2}\omega_{t}h)+x_{t}
\\
y_{t+h}:&=2\dfrac{v_{t}}{\omega_{t}}\sin(\dfrac{1}{2}\omega_{t}h)\sin(\theta_{t}+\dfrac{1}{2}\omega_{t}h)+y_{t}
\\
\theta_{t+h}:&=\omega_{t}h+\theta_{t},
\end{align*}
and for $\omega_{t}=0$, it is
\begin{align*}
x_{t+h}:&=v_{t}h\cos(\theta_{t})+x_{t}
\\
y_{t+h}:&=v_{t}h\sin(\theta_{t})+y_{t}
\\
\theta_{t+h}:&=\theta_{t}.
\end{align*}
For an HV:
\begin{align*}
x_{t+h}:&=\int_{0}^{h}v(t+\tau)\cos(\theta_{t+\tau})d\tau+x_{t}
\\
y_{t+h}:&=\int_{0}^{h}v(t+\tau)\sin(\theta_{t+\tau})d\tau+y_{t}
\\
\theta_{t+h}:&=\int_{0}^{h}\omega(t+\tau)d\tau+\theta_{t}.
\end{align*}

We allow different maximum achievable decelerations for AVs and HVs, i.e., different lower bounds in the physical constraints on acceleration, and use a superscript ``hv'' and ``av'' to differentiate them. Specifically, we allow $a^{hv}_{min} > a^{av}_{min}$. We also suppose that HVs may need longer response time $T^{hv}_{r}$ than AVs. We also suppose that AVs have the knowledge on the type of nearby vehicles e.g., AV or HV, as they can confirm the identity of each AV through V2V communication. HVs, in contrast, are given the freedom to treat each nearby vehicle as an HV. As a result each AV is expected to be responsible for not colliding with vehicles, and, at the same time, behaving in a manner similar to an HV when it is interacting with an HV. By imposing such concern on AVs, we ensure that HVs do not need to differentiate between the types of vehicles that are following it. However, an AV must ensure that if it is followed by an HV, then it should give enough response time to the following HV and also take into account that it can at most brake with $a^{hv}_{min}$ to avoid collision. The main contribution of this paper including a series of our previous studies \cite{c10_liu2015towards,c11_liu2016towards,hsieh2017throughput} is to show how to integrate such a loose model of HVs and a tight model of AVs safely in one system.

\section{Methods}\label{sec:methods}

We consider the intersection management problem for mixed traffic and establish its safety in this section. It is challenging to design a safe intersection management scheme for mixed traffic, due to the differences in control freedom, command complexity and collision-avoidance capabilities between AVs and HVs. Moreover, AVs are expected to follow instructions from the intersection while HVs are assumed to follow signal lights for motion planning. Also, AVs are able to behave like HVs and comply to signal lights as well.

\begin{figure}[thpb]
\centering
\includegraphics[width=.5\textwidth, height=.33\textwidth]{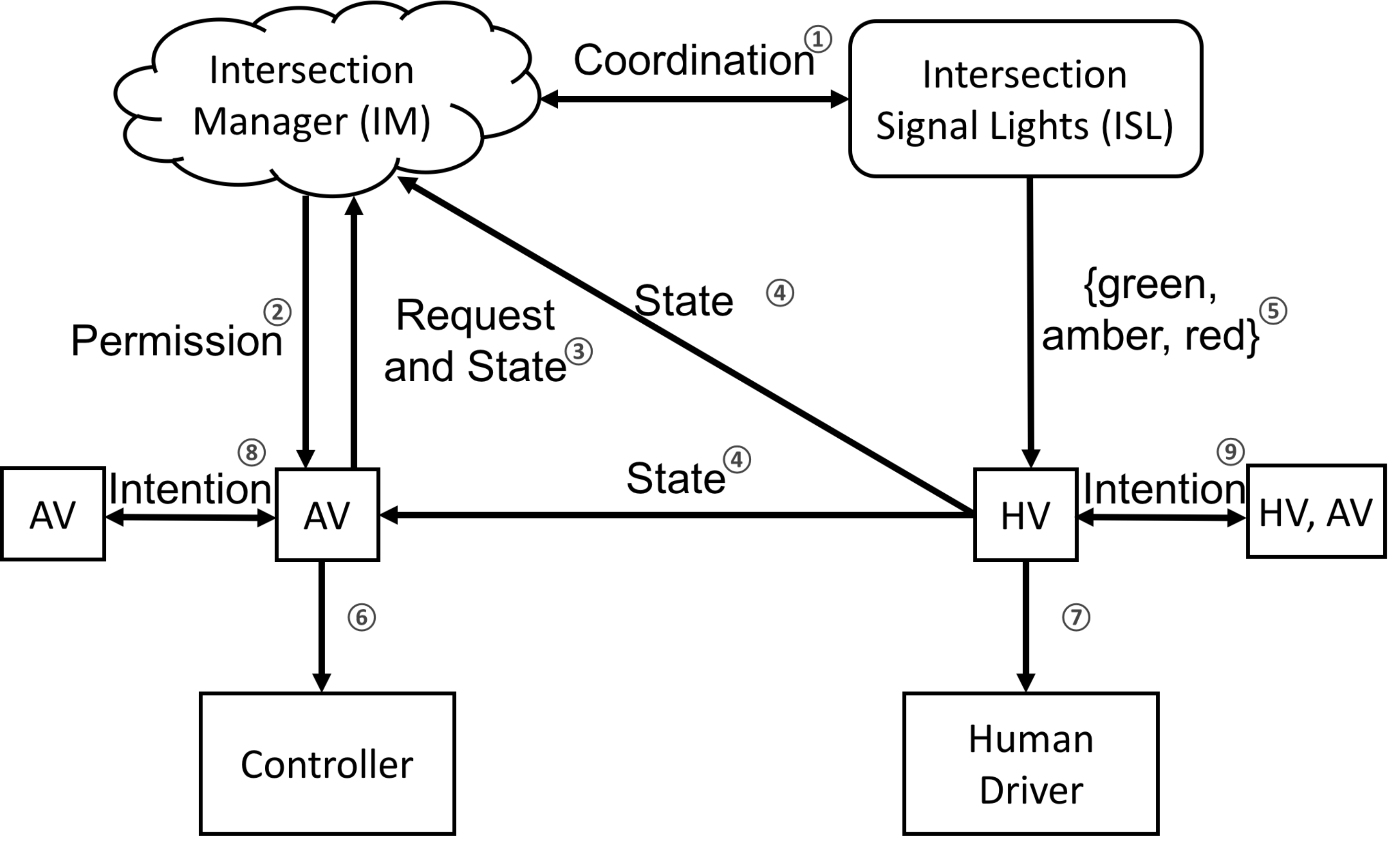}
\caption{Architecture of intersection management for mixed traffic.}
\label{fig:architecture}
\end{figure}

Fig.\ref{fig:architecture} illustrates the proposed architecture for intersection management in mixed traffic. The Intersection Manager (IM) is responsible for managing AVs such that they will not enter the collision area at the same time as other vehicles. It collects AV states and HV states through V2I communication and road-side sensors, after they are within $A_{C}$. Based on that information, IM determines whether or not to grant pass-through permission to an AV according to its ``permission assignment policy''. Similarly the Intersection Signal Light Controller (ISL) controls the signal light color according to a ``signal operation policy''. The color of the signal light is observed visually after an HV enters $A_{H}$.

\begin{table}[thpb]
\centering
\caption{Illustration of Marks in Fig.\ref{fig:architecture}}
{\small
\begin{tabular}{|p{0.6cm}|p{5cm}|}
\hline
Mark & Illustration
\\ \hline
\textcircled{1} & IM-ISL Coordination Protocol
\\ \hline
\textcircled{2} & Permission Assignment Policy
\\  \hline
\textcircled{3} & V2I Communication
\\ \hline
\textcircled{4} & Sensors
\\ \hline
\textcircled{5} & Signal Operation Policy
\\ \hline
\textcircled{6} & Motion Protocol
\\ \hline
\textcircled{7} & Simple Rules
\\ \hline
\textcircled{8} & V2V Communication, Sensors
\\ \hline
\textcircled{9} & Turning Lights
\\ \hline
\end{tabular}
}
\label{table:architecture}
\end{table}

Since AVs and HVs follow different management policies, but they interleave with each other along their paths, it is critical to have coordination between IM and ISL. Otherwise, an HV approaching the intersection may observe a green signal light, but be blocked by AVs ahead of it whose requests were denied. Also, the coordination needs to respect the uncertainty of HVs' understanding in signal lights as the communication between HVs and the intersection is unidirectional, i.e., HVs cannot directly send acknowledgement to or directly receive certain permissions from intersections. For instance, all HVs may observe a green signal light as they enter $A_{H}$, even though the ISL may only plan to permit a subset of them. When the signal light changes, the remaining HVs may fail to brake before $A_{I}$ and thus become ``unplanned but permitted'' vehicles. 

After receiving permission from the IM, AVs control their movement through their motion protocol. The target of the adjustment is to safely pass through the intersection. Meanwhile, AVs need to stay safe with respect to vehicles nearby by taking into account their states and intentions. The states and intention of nearby AVs are obtained through V2V communication and on-board sensors, while information of nearby HVs are learnt through on-board sensors. In contrast, HVs being operated by human drivers are not expected to follow as complex commands and tighter constraints as AVs, and obtain intentions of other vehicles through their turn signals. These features are illustrated in Fig. \ref{fig:architecture} through the marks explained in Tab \ref{table:architecture}.

\subsection{Safe Signal Operation for HVs}

To respect the capability differences between HVs and AVs, we propose a flexible traffic light operation approach, and, correspondingly, very simple rules for HVs. The rules are easy to follow and conform to the habits of human drivers in present traffic. Denote by $\Gamma$ the set of all paths in Fig. \ref{fig:intersection}, and by $\gamma_{i}$ the path with index $i$. We suppose that except the right-turn paths $\Gamma_{r} = \{\gamma_{1}, \gamma_{5}, \gamma_{9}, \gamma_{13}\}$, all other paths are managed by their corresponding signal lights. A signal light may be allowed to manage more than one path. For instance, $\gamma_{2}$ and $\gamma_{3}$ can be managed by one signal light. Denote by $\ell_{\gamma}(t)$ the status of signal light for path $\gamma$ at time $t$, with $\ell_{\gamma}(t) \in \{g, a, r\}$, where $g$ represents green, $a$ amber, and $r$ red. We use $\gamma_{i} \cap \gamma_{j}$ to indicate the collision area between paths $\gamma_{i}$ and $\gamma_{j}$. $\gamma_{i} \cap \gamma_{j} = \emptyset$ indicates $\gamma_{i}$ and $\gamma_{j}$ are collision-free, e.g., $\gamma_{2}$ and $\gamma_{10}$. Denote by $p_{\gamma}$ the cross-point of path $\gamma$ and $A_{I}$. If two vehicles $c_{i}$ and $c_{j}$ are on the same path $\gamma$ at time $t$, then the distance from $c_{i}$ to $c_{j}$ along is denoted by $d_{\gamma}(p_{t}(c_{i}), p_{t}(c_{j}))$. $d_{\gamma}(p_{t}(c_{i}), p_{t}(c_{j}))>0$ if $c_{i}$ is following $c_{j}$; otherwise $d_{\gamma}(p_{t}(c_{i}), p_{t}(c_{j}))<0$. 

Now we specify signal operation policy.
\begin{policy}
\textbf{Signal Operation Policy}

\noindent Intersection signal lights are operated in such a way that for any $\gamma$, $\ell_{\gamma}(t)$ cyclically changes from $g$ to $a$ to $r$ and back to $g$ as $t$ increases. Meanwhile, it respects two constraints:
\begin{enumerate}
    \item For any $\gamma_{i}$, if $\ell_{\gamma_{i}}(t) \in \{g, a\}$ then $\ell_{\gamma_{j}}(t)=r$ for all $\gamma_{j} \in \Gamma_{i}$, where $\Gamma_{i} = \{\gamma \in \Gamma: \gamma_{i} \cap \gamma \neq \emptyset \}$.That is, when a path is green or amber, all its conflicting paths are red.
    \item  for any $\gamma \in \Gamma$, suppose at some time $t_{0}$, there exists $\ell_{\gamma}(t_{0})=a$, then $\ell_{\gamma}(t) \neq r$ for 
    $t>t_{0}$ as long as some HV $c$ along $\gamma$ cannot stop before $A_{I}$, i.e. $d_{\gamma}(p_{t}(c), p_{\gamma})<s_{hv}(v_{t}(c), 0)$, where
    \begin{align*}
       s_{hv}(v_{t}(c), 0) = \dfrac{v_{t}(c)^2}{-2a^{hv}_{min}}+v_{max}T^{hv}_{r}+s_{min}.
    \end{align*}
     That is, a path is held at amber for sufficiently long time to allow HVs that cannot stop to pass through the intersection.
\end{enumerate}
\end{policy}
It is easy to find a feasible solution that satisfies constraint (1). In fact, most intersection signal lights in operation presently satisfy constraint (1). Constraint (2) is also easy to satisfy as the intersection knows HV states through road-side sensors.


Given the signal light state in the intersection, HVs are expected to follow simple rules within $A_{H}$. For simplicity, we assume that $d_{H}\geq s_{hv}(v_{max},0)$, i.e. if an HV observes a red light when it enters $A_{H}$ with velocity $v_{max}$, it is still able to stop before entering $A_{I}$. Denote by $\mathcal{C}^{hv}$ the set of all HVs and $\mathcal{C}^{av}$ the set of all AVs. Given vehicle $c_{i}$, we denote by $\gamma(c_{i})$ the path along with $c_{i}$ is moving. The set of vehicles that are ahead of $c_{i}$ on $\gamma(c_{i})$ at time $t$ can be expressed as :
\begin{align}
\mathcal{C}_{i}(t) :=\{c\in \mathcal{C}^{hv}\cup \mathcal{C}^{av}:\gamma(c)=\gamma(c_{i}) \nonumber \\
\land d_{\gamma(c)}(p_{t}(c_{i}),p_{t}(c))>0\}.
\end{align}
To establish system-wide safety within the intersection, we assume that HVs are able to follow their lead vehicles on the same path while maintaining a safe separation distance. This assumption is reasonable as otherwise, safety cannot be established even in an intersection with only HVs. We will show that it is enough for safety to follow the nearest lead vehicle along its path with safe separation distance.

Now we specify the rules for HVs within $A_{H}$.
\begin{rules}
\noindent\textbf{Rules for HVs}

\noindent A HV $c \in \mathcal{C}^{hv}$ follows the following rules after entering $A_{H}$: 
\begin{enumerate}
    \item $c$ follows the nearest lead vehicle $c^{l}$ maintaining a safe separation distance, $d_{\gamma(c)}(p_{t}(c), p_{t}(c^{l})) \geq s_{hv}(v_{t}(c),v_{t}(c^{l}))$ at any $t$, where
    \begin{align*}
        & s_{hv}(v_{t}(c),v_{t}(c^{l})) :=   \\
        &\bm{1}(v_{t}(c)>v_{t}(c^{l}))\dfrac{(v_{t}(c))^2-(v_{t}(c^{l}))^2}{-2a^{hv}_{min}}+\\
        & v_{max}T^{hv}_{r} + s_{min},
    \end{align*}
    and $s_{min}\geq 0$ is a constant to take care of vehicle size. The nearest lead vehicle is determined by
    \begin{align*}
        c^{l}:=&\arg\underset{c'\in 
        \mathcal{C}_{c}(t)}{\min}\,\{d_{\gamma(c)}(p_{t}(c),p_{t}(c')):\\
        &d_{\gamma(c)}(p_{t}(c),p_{t}(c'))>0\}.
    \end{align*}

\item If $\gamma(c)\notin \Gamma_{r}$, then $c$ prepares to enter $A_{I}$ if the following condition holds
\begin{align*}
    &\big(\ell_{\gamma(c)}(t)=g\big) \lor \big(\ell_{\gamma(c)}(t)=a  \\
    & \land d_{\gamma(c)}(p_{t}(c), p_{\gamma})<s_{hv}(v_{t}(c),0)\big),
\end{align*}
$c$ prepares to stop before $A_{I}$ if below condition holds
\begin{align*}
    & (\ell_{\gamma(c)}(t)=r) \lor \big(\ell_{\gamma(c)}(t)=a \nonumber \\
    & \land d_{\gamma(c)}(p_{t}(c), p_{\gamma})\geq s_{hv}(v_{t}(c),0)\big).
\end{align*}

\item If $\gamma(c)\in \Gamma_{r}$, then $c$ prepares to enter $A_{I}$ at time $t$ if 1) is satisfied and the following condition holds: denote by $\gamma$ the only conflicted path of $\gamma(c)$
\begin{align*}
    d_{\gamma}(p_{t}(c^{r}), p_{\gamma})\geq s_{hv}(v_{t}(c^{r})-a^{hv}_{min}h, 0).
\end{align*}
where $c^{r}$ is.
\begin{align*}
c^{r}:=&\arg{\underset{c'\in \mathcal{C}^{av}\cup \mathcal{C}^{hv}}\min}\,\{d_{\gamma}(p_{t}(c'),p_{\gamma}):\\
&d_{\gamma}(p_{t}(c'),p_{\gamma})>0\}.
\end{align*}
\end{enumerate}
\end{rules}
The first rule for HVs ensures that once the nearest lead vehicle brakes rapidly, an HV is able to stop before hitting the lead vehicle. To respect the difference in communication capabilities, the first rule does not differentiate between whether the lead vehicle is an HV or AV. The second rule governs how an HV behaves based on signal light status. We highlight here that when the signal light turns to amber, an HV will still plan to pass through the intersection if its current speed and position do not allow a safe stop before entering $A_{I}$; otherwise, it prepares to stop before entering $A_{I}$. The third rule governs how an HV turns right safely, when there is no signal light. Since a right-turn path has a collision area with only one another path, the third rule specifically checks other vehicles on that path. In view the fact that AVs may suddenly change speed, an adjustment of observed speed is added to the safe separation distance.

\subsection{Safe Motion Protocol for AVs}
As mentioned in Section \ref{sec:methods}, AVs are expected to send path requests to the IM after entering $A_{C}$. After receiving a decision from IM, AVs adjust their movement through their motion protocol. In this subsection, we first present the design of the AV controller and then propose a motion protocol for AVs to stay safe. We propose a Model Predictive Control (MPC) based motion planner for AVs in combination with safety constraints. Given $v_{t-h}$, the velocity of an AV in time slot $[t-h, t)$, the target of the MPC based controller is to decide control inputs for the time slot $[t, t+h)$. Then motion protocol adjusts the controller by imposing different safety constraint. Usually the safety constraints that maintains a safe separation distance from the nearest lead vehicle is imposed. If the motion protocol indicates that the AV should stop before entering $A_{I}$, then additional safety constraints that guide the AV to stop at $p_{\gamma}$ will be imposed in exchange. We will show that under the proposed architecture and protocols, it is enough for each AV to follow one possibly virtual ``reference object'' for safety. The ``reference object'' can be the nearest vehicle or the cross-point of the $\gamma(c)$ and $A_{I}$. Denote by $v^{ro}_{t}$ and $p^{ro}_{t}$ the sampled velocity and position of the reference object at time $t$. Motion protocol achieves the goal of adjusting controller majorly through changing the reference object. 

We employ the following MPC to determine the movement of an AV within $A_{C}$.

\leftline{\textbf{MPC for AVs within $A_{C}$}}
\vspace{-6mm}
\begin{align}
\underset{\mathbf{u}(0:N-1)}{\min}\,\hspace{2mm} &J(\mathbf{x_{t}}, \mathbf{x_{t}^{f}}, {\mathbf{u}(0:N-1)})\\
s.t.\hspace{2mm}
&\mathbf{x}_{t+(k+1)h}=f(\mathbf{x}_{t+kh}, \mathbf{u}_{t+kh})\nonumber\\ &d_{\gamma}(p_{t+(k+1)h},p_{t+(k+1)h}^{ro})\geqslant{s^{*}_{*}(v_{t+kh},v_{t+kh}^{ro})}\nonumber\\
&d_{\gamma}(p_{t+kh}^{ro},p_{t+(k+1)h}^{ro})=(v_{t+(k-1)h}^{ro}+a^{*}_{min}h)h\nonumber\\
&d_{\gamma}(p_{t+kh},p_{t+(k+1)h})=v_{t+kh}h\nonumber\\
&\omega_{t+kh}=v_{t+kh}/\rho_{t+kh}\nonumber\\
&v_{t+kh}\in[{\underline{v}_{t+kh}}, \overline{v}_{t+kh}]\nonumber\\
&\underline{v}_{t+kh}:=\max\{0, v_{t+(k-1)h}+a^{*}_{\min}h\} \nonumber \\
&\overline{v}_{t+kh}:=\min\{v_{max},v_{t+(k-1)h}+a_{max}h\} \nonumber \\
&\theta_{min}\leq\theta_{t+(k+1)h}\leq\theta_{max}\nonumber,
\end{align}
for all $k\in\{0,...,N-1\}$, where $\mathbf{u}(0:N-1)=\{\mathbf{u}_{t},...,\mathbf{u}_{t+(N-1)h}\}$, $N$ is the length of the time horizon, and $\mathbf{x}^{f}_{t+kh}$ is the target state set for time $t+kh$. The cost function $J$ is allowed to be arbitrary, since our guarantee of safety depends only on the existence of a feasible solution to the MPC, and not on the objective function or its value. The separation distance $s_{*}^{*}(\cdot,\cdot)$ for AVs is also designed to ensure that, even if its lead vehicle brakes rapidly, it can stop before hitting its lead vehicle.

The motion protocol needs to ensure the safety of AVs in the multiple scenarios shown in Fig. \ref{fig:mixed_traffic_scenarios}. Since we do not require that an HV differentiates between the type of its lead vehicle and do not require that it takes into account any vehicles that are following it, thus we enforce that an AV to behave like an HV (not decelerate more than $a^{hv}_{min}$) if it is followed by an HV. Such an AV is referred to as a ``virtual HV" (VHV) and is treated as HV in $s^{*}_{*}(\cdot,\cdot)$. Thus, as shown in Fig. \ref{fig:mixed_traffic_scenarios}(a,b), $a^{*}_{min}=a^{hv}_{min}$ when an AV is followed by an HV/VHV, and, as shown in Fig. \ref{fig:mixed_traffic_scenarios}(c,d) $a^{*}_{min}=a^{av}_{min}$ when an AV is followed by another AV. At the same time, however, the AV may be following yet another AV that can decelerate at $a^{av}_{min}$. Therefore, $s^{*}_{*}(\cdot,\cdot)$ the lower bound on distance between an AV and its lead vehicle depends on the type of vehicles the AV is following (specified by the superscript of $s^{*}_{*}(\cdot,\cdot)$), as well as the type of vehicle it is followed by (specified by the subscript of $s^{*}_{*}(\cdot,\cdot)$). Covering all these cases, all possible $s^{*}_{*}(\cdot,\cdot)$ are given below:
\begin{align}
s^{av}_{av}(v_{t-h},v_{t-h}^{ro})&:=\bm{1}(v_{t-h}>v_{t-h}^{ro})\big\{\dfrac{(v_{t-h})^2-(v_{t-h}^{ro})^{2}}{-2a^{av}_{min}}\nonumber\\
&+(v_{t-h}-v_{t-h}^{ro})h-\dfrac{1}{2}a^{av}_{min}h^{2}\big\}+s_{min},\\
s_{av}^{hv}(v_{t-h}, v^{ro}_{t-h})&:=\bm{1}(v_{t-h}>v_{t-h}^{ro}+a^{hv}_{min}h/2)\big\{\nonumber\\
&\dfrac{(v_{t-h}-v_{t-h}^{ro}-a^{hv}_{min}h/2)^{2}}{-2(a^{av}_{min}-a^{hv}_{min})}+(v_{t-h}-v_{t-h}^{ro})h
\nonumber \\
&-\dfrac{1}{2}a^{av}_{min}h^{2}-\dfrac{1}{2}a^{hv}_{min}h^{2}\big\}+s_{min},\\
s_{hv}^{hv}(v_{t-h}, v^{ro}_{t-h})&:=\bm{1}(v_{t-h}>v_{t-h}^{ro}+a^{hv}_{min}h/2)\big\{\nonumber \\
&\dfrac{v_{t-h}^{2}-(v_{t-h}^{ro}+a^{hv}_{min}h/2)^{2}}{-2a^{hv}_{min}}+(v_{t-h}-v_{t-h}^{ro})h
\nonumber \\
&-a^{hv}_{min}h^{2}\big\}+v_{max}T_{r}^{hv}+s_{min},\\
s_{hv}^{av}(v_{t-h}, v^{ro}_{t-h}&):=\bm{1}(v_{t-h}>\sqrt{\dfrac{a^{hv}_{min}}{a^{av}_{min}}}v_{t-h}^{ro})\big\{\dfrac{v_{t-h}^{2}}{-2a^{hv}_{min}}-\nonumber \\
&\dfrac{(v_{t-h}^{ro})^{2}}{-2a^{av}_{min}}+(v_{t-h}-v_{t-h}^{ro})h-\dfrac{1}{2}a^{hv}_{min}h^{2}\big\} \nonumber \\
&+v_{max}T_{r}^{hv}+s_{min}.
\end{align}
We observe that different from \cite{c6_kim2014mpc}, whose safe separation distance grows polynomially with velocity $d_{s}(v_{t-h})=v_{t-h}^2/(-2a^{av}_{min})+v_{t-h}h-a^{av}_{min}h^{2}/2>v_{t-h}^2/-2a^{av}_{min}$, our separation distance $s^{av}_{av}(v_{t-h},v^{ro}_{t-h})$ can be as small as $s_{min}$, independent of velocity. This property makes the MPC amenable to platooning \cite{c10_liu2015towards}. Actually $d_{s}(v_{t-h})$ can be viewed as a special case of $s^{av}_{av}(v_{t-h},v^{ro}_{t-h})$ when $v_{t-h}^{ro}=0$.

\begin{figure}[thpb]
\centering
\includegraphics[width=.5\textwidth, height=.39\textwidth]{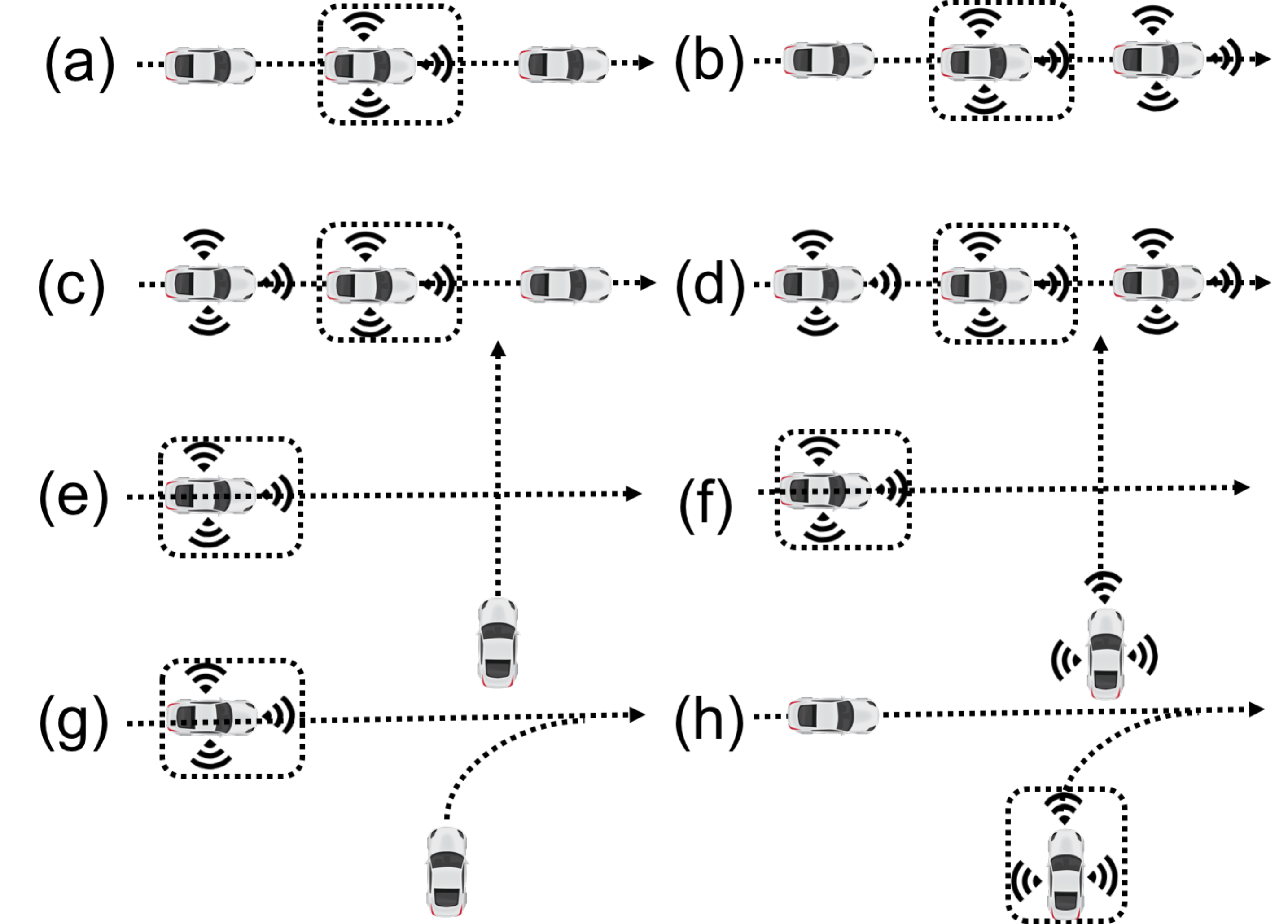}
\caption{Multiple traffic scenarios. An AVis represented by an icon with a WiFi symbol. The vehicle with dashed square is ``own vehicle" whose viewpoint and behavior are discussed in the test.}
\label{fig:mixed_traffic_scenarios}
\end{figure}

We assume reasonably that $d_{C}\geq v_{max}^{2}/-2a^{hv}_{min}+v_{max}h-a^{hv}_{min}h^{2}/2+v_{max}T^{hv}_{r}+s_{min}$, i.e., if an AV that is followed by an HV/VHV enters $A_{C}$ with velocity $v_{max}$, it still can stop before entering $A_{I}$. Otherwise, it is impossible for those AVs. We observe that when an AV is following an HV, $v_{t-h}^{ro}$ is the momentary velocity of its lead HV at $t-h$, as HVs enjoy the freedom to change their control input ay any time during $[t-h, t)$. Consequently in the $t$-round MPC for that AV, the velocity of lead HV is taken as $v^{ro}_{t-h}+a^{hv}_{min}h/2$ in equations (5-6), i.e., the estimated average velocity for time $[t-h,t)$. We also observe that although the result from the $t$-th epoch MPC (the MPC whose initial time is $t$) is a plan for a longer time horizon, we only implement it for the first time slot $[t,t+h)$. 

\begin{figure}[thpb]
\centering
\includegraphics[width=.5\textwidth, height=.24\textwidth]{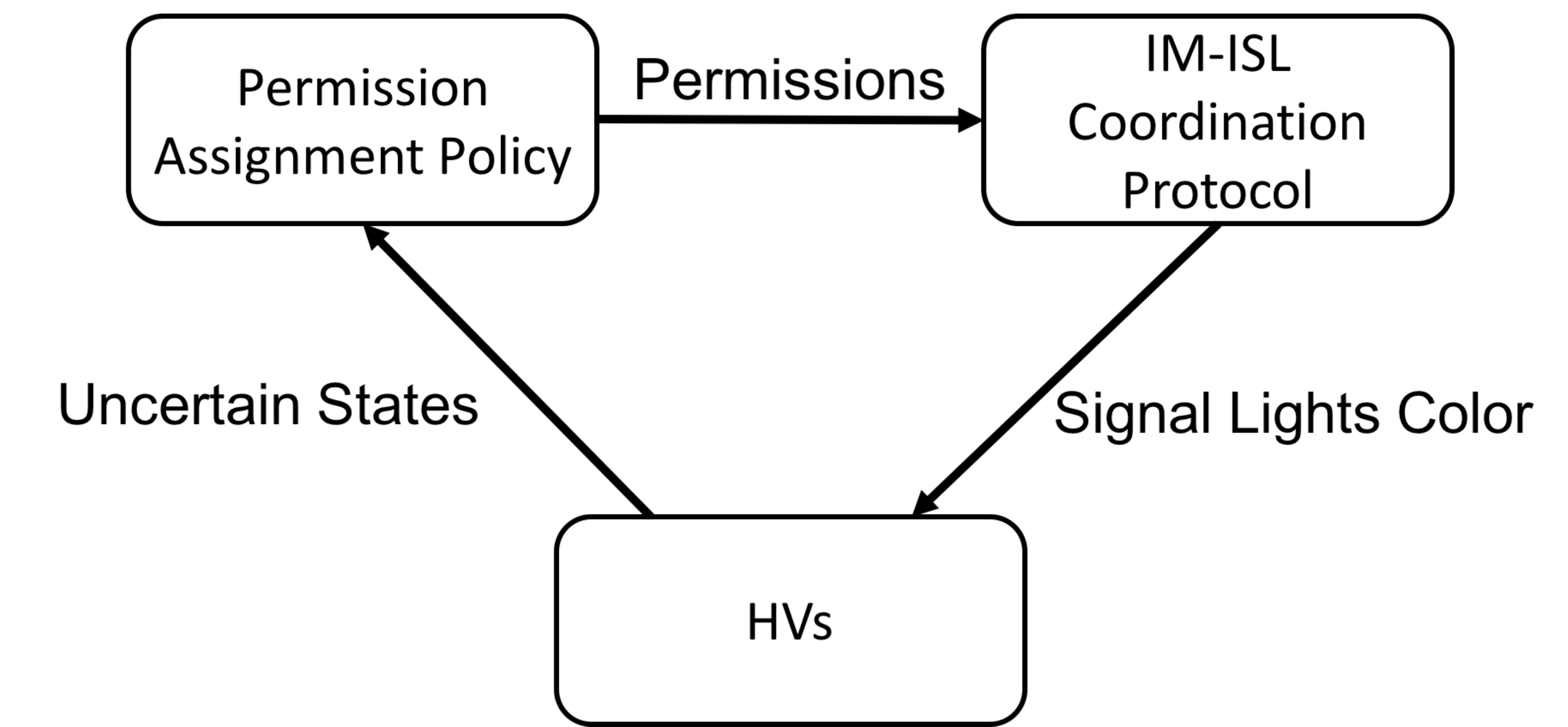}
\caption{Information flow between permission assignment and signal color determination.}
\label{fig:decision_flow}
\end{figure}

\vspace{10mm}

The motion protocol for AVs is given below.
\begin{protocol} \textbf{Motion Protocol for AVs}

\noindent An AV $c\in \mathcal{C}^{av}$ follows the following motion protocol to adjust its MPC based controller:
\begin{enumerate}
    \item If the request of $c$ is permitted by the IM, then $c$ follows the nearest lead vehicle $c^{ro}$ with safe separation distance $s^{*}_{*}(v_{t-h}(c),v_{t-h}(c^{ro}))$ in equations (4-7) depending on which scenario $c$ is in. $c^{ro}$ is determined through
    \begin{align*}
        c^{ro}:=&\arg\underset{c'\in 
        \mathcal{C}_{c}(t)}{\min}\,\{d_{\gamma(c)}(p_{t}(c),p_{t}(c')):\\
        &d_{\gamma(c)}(p_{t}(c),p_{t}(c'))>0\}.
    \end{align*}
    
    \item If the request of $c$ has not been permitted by the IM, then $c$ maintains a safe separation distance from its nearest lead vehicle while preparing to stop before entering $A_{I}$. Correspondingly, we add one additional constraints to Equation (3):
    \begin{align*}
        d_{\gamma}(p_{t+(k+1)h},p_{\gamma})\geqslant{s^{*}_{*}(v_{t+kh},0)}.
    \end{align*}
\end{enumerate}
\end{protocol}

\subsection{Safe Intersection Management for a Mix of AVs and HVs}

Lacking communication capability and ability to follow complex commands are the major issues for HVs. Information can only be communicated to HVs through traffic lights, and no information or acknowledgements can be obtained from them.  This greatly increases the uncertainty surrounding HVs in traffic management. In contrast, as shown in Fig. \ref{fig:architecture}, IM taken on their request to pass through the intersection, can collect acknowledgements from AVs. Handling the uncertainty vis-vis HVs impairs the throughput of intersections, since when an intersection is not sure about an HV's plan, it needs to make a worst estimate to guarantee safety. The objective our traffic management scheme is to guarantee safety of mixed traffic while simultaneously achieving greater throughput as possible.

As discussed in Section \ref{sec:problem_modeling}, AVs send their state and a request to pass-through after they enter $A_{C}$. The received vehicle state for AVs consists of their position information $p_{t}$, and control inputs for the last time slot $\mathbf{u}_{t-h}$. At the same time, an HV's state is obtained through roadside sensors. Paths to follow for HVs are obtained through their turn signaling. Given this totality of information, the IM permits path requests of AVs by following the permission assignment policy. Considering that processing requests from AVs takes time, we suppose that the IM processes the requests received during $[t-h, t)$ at $t$. Permissions of those request are sent to corresponding AVs after $t$. Denote by $\mathcal{P}^{av}_{t}$ the set of AVs that the IM grants permissions to before time $t$. Denote by $\mathcal{P}^{hv}_{t}$ the set of HVs IM plans to grant permission before time $t$. This means that the IM-ISL coordination protocol will enforce the corresponding signal light to be green until some permitted HV enters $A_{I}$.

The process above will introduce uncertainty as other HVs, on the same lane may think they are permitted to enter as well and may not be able to safely stop before $A_{I}$ after the light changes. The IM therefore needs to take that uncertainty into account. Denote by $\mathcal{C}^{un}_{t}$ an estimate of those ``uncertain'' HVs; then:
\begin{align}
    &\mathcal{C}^{un}_{t} := \{c\in\mathcal{C}^{hv}: \ell_{\gamma(c)}(t-h)=g  \ \land d_{\gamma(c)}(p_{t}(c), p_{\gamma(c)})\nonumber \\
    &<s_{hv}(\max(v_{t-h}(c)+a_{max}h,v_{max}), 0)\}.
\end{align}
$\mathcal{C}^{un}_{t}$ contains those HVs that may be unable to stop before $A_{I}$ in the next time slot $[t,t+h)$ once their corresponding signal light turns to amber. In the definition of $\mathcal{C}^{un}_{t}$, we conservatively use $\max(v_{t-h}(c)+a_{max}h,v_{max})$ since the estimated present velocity as that is the maximum achievable velocity for the HV at time $t$. Although $\mathcal{C}^{un}_{t}$ overestimates the set of uncertain HVs and may impair throughput, it is necessary to guarantee safety. Denote by $\mathcal{C}^{e}_{t}$ the set of vehicles that exited $A_{I}$ before $t$; then we have:
\begin{align}
    \mathcal{C}^{e}_{t} :=&\{c\in\mathcal{C}: d_{\gamma(c)}(p_{t}(c),p_{\gamma(c)})<0 \ \land \nonumber\\
    &p_{t}(c)\notin A_{I}\}.
\end{align}

Now we describe the permission assignment policy below.
\begin{policy} \textbf{Permission Assignment Policy}

\noindent Suppose that AVs send requests for paths to the IM once they enter $A_{C}$. At time $t$, given $\mathcal{P}^{av}_{t-h}$,  $\mathcal{P}^{hv}_{t-h}$, $\mathcal{C}^{un}_{t}$, and additional path requests received during $[t-h, t)$, the IM follows the following steps to determine $\mathcal{P}^{av}_{t}$,  $\mathcal{P}^{hv}_{t}$.

\begin{enumerate}
    \item Initialize by still permitting already-permitted vehicles and remove those that already exited:
    \begin{align*}
        \mathcal{P}^{av}_{t} = \mathcal{P}^{av}_{t-h} \setminus \mathcal{C}^{e}_{t},\\
        \mathcal{P}^{hv}_{t} = \mathcal{P}^{hv}_{t-h} \setminus \mathcal{C}^{e}_{t}.
    \end{align*}

    \item Sort pending requests (including new requests received during $[t-h, t)$ and un-permitted requests until $t$) based on the present distance from vehicles to the intersection: the closer to intersection, the lower the requests rank.
    
    \item Process the requests from low rank to high rank: for an AV $c$ with a request for path $\gamma(c)$, if $\gamma(c) \cap \gamma(c') = \emptyset$ for any $c' \in \mathcal{P}^{av}_{t} \cup \mathcal{P}^{hv}_{t} \cup \mathcal{C}^{un}_{t}$, then:
    \begin{align*}
        \mathcal{P}^{av}_{t} = \mathcal{P}^{av}_{t} \cup c.
    \end{align*}
    
    \item For each HV $c\in \mathcal{C}^{uv}$, $\mathcal{P}^{hv}_{t} = \mathcal{P}^{hv}_{t} \cup c$ as long as one of the following conditions holds: (\emph{a}) $\gamma(c')\cap \gamma(c)=\emptyset$ for any $c' \in \mathcal{P}^{av}_{t} \cup \mathcal{P}^{hv}_{t} \cup \mathcal{C}^{un}_{t}$; (\emph{b}) there exists  $c^{f}\in \mathcal{P}^{av}_{t} \cup \mathcal{C}^{un}_{t}$, where $c^{f}$ satisfies:
    \begin{align*}
        &\gamma(c) =\gamma(c^{f}),\\
        &d_{\gamma(c)}(p_{t}(c),p_{t}(c^{f}))
        <0.
    \end{align*}
\end{enumerate}
\end{policy}

The first step in the permission assignment policy ensures once IM assigns a permission to a vehicle for some time slot, it will keep permitting the vehicle in the following time slots until the vehicle exits $A_{I}$. The second step guarantees that the vehicles closer to $A_{I}$ enjoy higher priority. This is to improve management efficiency as and reduce delay. The third step indicates IM only grants permission to a request when the requested path will not collide with any previously permitted requests as well as uncertain HVs. This is the major difference between the proposed approach in this paper and previous studies, i.e., uncertainty of HVs is considered in management based on HV dynamics and physical constraints. In the fourth step, IM will plan to permit HVs if they are followed by some permitted AVs or their paths are collision-free with all permitted paths. We will discuss how to take use of this in IM-ISL coordination protocol.

After the permission assignment policy determines $\mathcal{P}_{t}^{hv}$, it needs to notify the relevant HVs through traffic signal lights. Therefore, an approach to determine the color of signal lights based on the outputs of permissions assignment is necessary. This is handled by IM-ISL coordination protocol. After the planned permission to HVs are implicitly declared through signal lights, more uncertainty will emerge since HVs that are not planned to permit by the IM may misunderstand the green lights as a signal of permission. This uncertainty will be handled in the next-round of permission assignment for system safety. The entire interaction process is presented in Fig. \ref{fig:decision_flow}. The IM-ISL Coordination Protocol is presented below.

\begin{protocol}
\textbf{IM-ISL Coordination Protocol}

\noindent Given $\mathcal{P}^{hv}_{t}$ and $\mathcal{C}^{un}_{t}$, signal light colors are determined by the following protocol, for any $\gamma \in \Gamma$:

\begin{enumerate}
    \item $\ell_{\gamma}(t)$ cyclically changes to $g$, $a$, and $r$  as $t$ increases.
    
    \item $\ell_{\gamma}(t)=g$ if there exists a vehicle $c$ such that $\gamma(c)=\gamma$ and $c\in P_{t}^{hv}$.
    
    \item $\ell_{\gamma}(t)$ changes to $a$ if there exists a vehicle $c$ with $\gamma(c)=\gamma$ and $c\in C_{t}^{un}\setminus P_{t}^{hv}$.

    \item $\ell_{\gamma}(t)=r$ if for any HV $c$ that has $\gamma(c)=\gamma$, there exists $c\notin C_{t}^{un} \cup P_{t}^{hv}$.
\end{enumerate}
\end{protocol}

For an HV $c$, represent the relationship that $c$ maintains a safe separation distance from the nearest vehicle at time $t_{0}$ by $c\mathcal{S}_{t_{0}}c^{l}$, i.e., $d_{\gamma(c)}(p_{t_{0}}(c), p_{t_{0}}(c^{l})) \geq s_{hv}(v_{t_{0}}(c),v_{t_{0}}(c^{l}))$.

We now outline the results that provide the safety guarantees.

\begin{lemma}
Suppose at initial time $t_{0}$, there are no HVs inside $A_{H}$. For each HV $c$ about to enter $A_{H}$ and its nearest lead vehicle $c^{l}$, one has $c\mathcal{S}_{t_{0}}c^{l}$. If all vehicles are HVs, the intersection signal lights are operated by Policy 1, and HVs move under Rule 1, then all HVs are safe at any time after $t_{0}$, i.e., $c\mathcal{S}_{t}c^{l}$ holds for every HV $c$ and any $t>t_{0}$.
\end{lemma}

\begin{proof}
As all vehicles are HVs, Rule 1.1 ensures $d_{\gamma(c)}(p_{t}(c), p_{t}(c^{l})) \geq s_{hv}(v_{t}(c),v_{t}(c^{l}))$ at any $t>t_{0}$. It is easy to verify that $s_{hv}(\cdot,\cdot)$ is enough for the following vehicle to stop before hitting the lead vehicle if the lead vehicle brakes with $a^{hv}_{min}$. If $\gamma(c) \in \Gamma_{r}$, since $d_{H}>s_{hv}(v_{max},0)$, $c$ is able to stop before entering $A_{I}$ safely. Rule 1.3 guarantees that $c^{r}\mathcal{S}_{t}c$ and $c\mathcal{S}_{t}c^{l}$. If $\gamma(c) \in \Gamma\setminus\Gamma_{r}$, Policy 1.1 guarantees that at any time, there will only be collision-free paths that are permitted. Therefore, under the constraints of the lemma, we have $c\mathcal{S}_{t}c^{l}$ holds for any HV $c$ and any $t>t_{0}$.
\end{proof}

For an AV $c$, represent the relationship that $c$ maintains a safe separation distance from its reference object at time $t_{0}$ by $c\mathcal{S}_{t_{0}}c^{ro}$, i.e., $d_{\gamma(c)}(p_{t_{0}}(c), p_{t_{0}}^{ro}) \geq s^{*}_{*}(v_{t_{0}-h},v_{t_{0}-h}^{ro})$.
\begin{lemma}
Suppose at initial time $t_{0}$, there are no AVs within $A_{C}$. For each AV $c$ about to enter $A_{C}$ and its reference object $c^{ro}$, one has $c\mathcal{S}_{t_{0}}c^{ro}$. The AVs are safe at every time $t>t_{0}$, if AVs adjust their MPC controller through Protocol 1 and permission assignment is collision-free, i.e., for any $c$, if $c$ receives permission, $\gamma(c)\cap \gamma(c')=\emptyset$ for each $c'$ that also receives permission.
\end{lemma}

\begin{proof}
This lemma considers a mixed traffic scenario. First, the permission assignment is assumed to be collision-free. Under that assumption, Protocol 1. handles the mixed traffic scenarios in Fig. \ref{fig:mixed_traffic_scenarios}. Suppose the initial condition $d_{\gamma(c)}(p_{t_{0}}(c), p_{t_{0}}^{ro}) \geq s^{*}_{*}(v_{t_{0}-h},v_{t_{0}-h}^{ro})$ is satisfied. Since $v^{ro}_{t_{0}}>max\{0,v^{ro}_{t_{0}-h}+a^{*}_{min}h\}$, there always exists $d_{\gamma(c)}(p_{t_{0}+h}(c), p_{t_{0}+h}^{ro}) \geq s^{*}_{*}(v_{t_{0}-h}+a^{*}_{min}h,v_{t_{0}}^{ro})$. That means that once the initial conditions are satisfied, deceleration with $a^{*}_{min}$ is always a feasible solution for MPC problem (3) for any AV. Therefore, $c\mathcal{S}_{t_{0}}c^{ro}$ indicates $c\mathcal{S}_{t}c^{ro}$ for any $t>t_{0}$.
\end{proof}

\begin{lemma}
Suppose we follow Policy 2 to assign permissions for requests, $\mathcal{P}^{hv}_{t}\cup\mathcal{P}^{av}_{t}\cup \mathcal{C}^{un}_{t}$ are collision-free, i.e., for any $c_{i}$, $c_{j}$ in $\mathcal{P}^{hv}_{t}\cup\mathcal{P}^{av}_{t}\cup \mathcal{C}^{un}_{t}$, one has $\gamma_{c_{i}}\cap \gamma_{c_{j}}=\emptyset$.
\end{lemma}

\begin{proof}
As in Step 3) and Step 4) of Policy 2, before providing permission to a new path, the IM will always check whether the path conflicts with any already-permitted paths.
\end{proof}

\begin{lemma}
Suppose we follow Protocol 2 to determine the colors of the intersection signal lights. Then the way that signal lights are operated satisfies all constraints in Policy 2.
\end{lemma}
 
\begin{proof}
Protocol 2.1 also lets $\ell_{\gamma}(t)$ cyclically change to $g$, $a$, and $r$ as $t$ increases. Meanwhile, from Protocol 2.2-2.3, paths whose signal lights are green or amber are all from $\mathcal{P}^{hv}_{t}\cup \mathcal{C}^{un}_{t}\subset \mathcal{P}^{hv}_{t}\cup \mathcal{P}^{av}_{t}\cup \mathcal{C}^{un}_{t}$, which are collision-free as noted in Lemma 3. That means they satisfy the constraint of Policy 2.1. In addition, from Protocol 2.3, $\mathcal{C}^{un}_{t}\setminus\mathcal{P}^{hv}_{t}$ contains those HVs that may fail to stop safely before $A_{I}$ based on worst estimation. In fact, it is a tighter constraints than Policy 2.2. Thus the way that signal lights are operated satisfies all constraints of Policy 2.
\end{proof}

\vspace{10mm}

\begin{theorem}
Suppose at initial time $t_{0}$, there are no vehicles inside $A_{C}$. Then for each HV $c$ about to enter $A_{H}$ and the its nearest lead vehicle $c^{l}$, one has $c\mathcal{S}_{t_{0}}c^{l}$. For each AV $c'$ about to enter $A_{C}$ and its reference object $c^{ro}$, there exists $c'\mathcal{S}_{t_{0}}c^{ro}$. If (1) all HVs move under Rules 1, (2) AVs adjust their MPC controller through Protocol 1, (3) the permission assignments follow Policy 2, and (4) the colors of traffic signal lights are determined through Protocol 2, then all HVs and AVs are safe at every time after $t_{0}$, i.e., $c\mathcal{S}_{t}c^{l}$ and $c'\mathcal{S}_{t}c^{ro}$ holds for any HV $c$ and any AV $c'$ at any $t>t_{0}$.
\end{theorem}

\begin{proof}
As indicated by Lemma 1, if all other vehicles move like HVs, the HVs can stay safe by following Rule 1 while the signal lights operation respects the constraints from Policy 1. Meanwhile, it has been shown in Lemma 2 that AVs are able to maintain a safe separation distance from their lead vehicles through Protocol 2, and behave like HVs when they are followed by HVs. As long as the permission assignments are collision free, then AVs are safe. Lemmas 3 and 4 directly demonstrate that the permission assignments (Policy 2) for AVs is collision-free and, at the same time, the signal lights operated through Protocol 2 satisfy the constraints in Policy 1. Therefore, the whole system is safe as long as all the assumptions in the theorem are valid.  
\end{proof}





\section{Concluding Remarks}\label{sec:conclusion}

Motivated by the fact that there will be a transition period when both human-driven vehicles and autonomous vehicles coexist in intersections, we have addressed the problem of design of provably safe intersection management for mixed transportation systems. The approach we follow considers a loose model of human behavior that permits worst case behavior, and a tight models for autonomous vehicles. We design a MPC controller for AVs and simple rules for HVs. We propose an architecture for intersection management consisting of permission assignments and traffic lights operation. We also design coordination protocols between permission assignment for AVs and color change of traffic lights for HVs. In the proposed approach, we take into account the differences between AVs and HVs in control freedom, command complexity, braking response and communication capability. To be compatible with human expectations and behavior, the rules and signal lights operation for HVs are close to what is currently used in traffic. Our work can be extended to more general scenario where vehicles with heterogeneous levels of autonomy as well as resultant uncertainty are operated. For instance, vehicles that can measure states through sensors, vehicles that can communicate through V2I and V2V communication, vehicles that can be directly controlled by intersections etc.


\vspace{10mm}

\section*{Acknowledgment}

This material is based upon work partially supported by NSF under
Contract Nos. CNS-1646449, CCF-1619085 and Science \& Technology Center Grant
CCF-0939370, Office of Naval Research under Contract N00014-18-1-2048, the U.S. Army Research Office under Contract No. W911NF-15-1-0279, and NPRP grant NPRP 8-1531-2-651 from the Qatar National
Research Fund, a member of Qatar Foundation.

\bibliographystyle{IEEEtran}
\bibliography{ICDCS_distributed}

\begin{thebibliography}{10}
\providecommand{\url}[1]{#1}
\csname url@samestyle\endcsname
\providecommand{\newblock}{\relax}
\providecommand{\bibinfo}[2]{#2}
\providecommand{\BIBentrySTDinterwordspacing}{\spaceskip=0pt\relax}
\providecommand{\BIBentryALTinterwordstretchfactor}{4}
\providecommand{\BIBentryALTinterwordspacing}{\spaceskip=\fontdimen2\font plus
\BIBentryALTinterwordstretchfactor\fontdimen3\font minus
  \fontdimen4\font\relax}
\providecommand{\BIBforeignlanguage}[2]{{%
\expandafter\ifx\csname l@#1\endcsname\relax
\typeout{** WARNING: IEEEtran.bst: No hyphenation pattern has been}%
\typeout{** loaded for the language `#1'. Using the pattern for}%
\typeout{** the default language instead.}%
\else
\language=\csname l@#1\endcsname
\fi
#2}}
\providecommand{\BIBdecl}{\relax}
\BIBdecl

\bibitem{c7_broughton2013traffic}
J.~Broughton, P.~Thomas, A.~Kirk, L.~Brown, G.~Yannis, P.~Evgenikos,
  P.~Papantoniou, N.~Candappa, M.~Christoph, K.~van Duijvenvoorde
  \emph{et~al.}, ``Traffic safety basic facts 2012: Junctions,'' 2013.

\bibitem{c8_nhtsa}
\BIBentryALTinterwordspacing
NHTSA. (2011) Fatality analysis reporting system (fars). [Online]. Available:
  \url{http://www.nhtsa.gov/FARS}
\BIBentrySTDinterwordspacing

\bibitem{c1_darpa2007}
\BIBentryALTinterwordspacing
DARPA. (2007) Darpa urban challenge. [Online]. Available: \url{http:
  //archive.darpa.mil/grandchallenge/}
\BIBentrySTDinterwordspacing

\bibitem{c2_ieee2013ieee}
I.~S. Association \emph{et~al.}, ``Ieee guide for wireless access in vehicular
  environments (wave) architecture,'' \emph{IEEE Std}, pp. 1609--0, 2013.

\bibitem{c3_dsrc2009dedicated}
I.~P802.11p/D1.1-2005, ``Draft amendment to standard for information technology
  - telecommunications and information exchange between systems - lan/man
  specific requirements,'' \emph{IEEE P802. 11p/D1. 1}, 2005.

\bibitem{c4_ncsl}
\BIBentryALTinterwordspacing
NCSL. (2018) Autonomous self-driving vehicles legislation. [Online]. Available:
  \url{http://www.ncsl.org/research/transportation/autonomous-vehicles-legislation.aspx}
\BIBentrySTDinterwordspacing

\bibitem{c17_kowshik2011provable}
H.~Kowshik, D.~Caveney, and P.~Kumar, ``Provable systemwide safety in
  intelligent intersections,'' \emph{IEEE transactions on vehicular
  technology}, vol.~60, no.~3, pp. 804--818, 2011.

\bibitem{c18_carlino2013auction}
D.~Carlino, S.~D. Boyles, and P.~Stone, ``Auction-based autonomous intersection
  management,'' in \emph{Intelligent Transportation Systems-(ITSC), 2013 16th
  International IEEE Conference on}.\hskip 1em plus 0.5em minus 0.4em\relax
  IEEE, 2013, pp. 529--534.

\bibitem{bashiri2017platoon}
M.~Bashiri and C.~H. Fleming, ``A platoon-based intersection management system
  for autonomous vehicles,'' in \emph{Intelligent Vehicles Symposium (IV), 2017
  IEEE}.\hskip 1em plus 0.5em minus 0.4em\relax IEEE, 2017, pp. 667--672.

\bibitem{c13_kamal2015vehicle}
M.~A.~S. Kamal, J.-i. Imura, T.~Hayakawa, A.~Ohata, and K.~Aihara, ``A
  vehicle-intersection coordination scheme for smooth flows of traffic without
  using traffic lights.'' \emph{IEEE Trans. Intelligent Transportation
  Systems}, vol.~16, no.~3, pp. 1136--1147, 2015.

\bibitem{qian2015decentralized}
X.~Qian, J.~Gregoire, A.~De~La~Fortelle, and F.~Moutarde, ``Decentralized model
  predictive control for smooth coordination of automated vehicles at
  intersection,'' in \emph{Control Conference (ECC), 2015 European}.\hskip 1em
  plus 0.5em minus 0.4em\relax IEEE, 2015, pp. 3452--3458.

\bibitem{c14_makarem2013model}
L.~Makarem and D.~Gillet, ``Model predictive coordination of autonomous
  vehicles crossing intersections,'' in \emph{Intelligent Transportation
  Systems-(ITSC), 2013 16th International IEEE Conference on}.\hskip 1em plus
  0.5em minus 0.4em\relax IEEE, 2013, pp. 1799--1804.

\bibitem{c15_de2013autonomous}
G.~R. de~Campos, P.~Falcone, and J.~Sjoberg, ``Autonomous cooperative driving:
  a velocity-based negotiation approach for intersection crossing,'' in
  \emph{Intelligent Transportation Systems-(ITSC), 2013 16th International IEEE
  Conference on}.\hskip 1em plus 0.5em minus 0.4em\relax IEEE, 2013, pp.
  1456--1461.

\bibitem{c9_de2014network}
A.~De~La~Fortelle, X.~Qian, S.~Diemer, J.~Gr{\'e}goire, F.~Moutarde,
  S.~Bonnabel, A.~Marjovi, A.~Martinoli, I.~Llatser, A.~Festag \emph{et~al.},
  ``Network of automated vehicles: the autonet 2030 vision,'' in \emph{ITS
  World Congress}, 2014.

\bibitem{c10_liu2015towards}
X.~Liu, K.~Ma, and P.~Kumar, ``Towards provably safe mixed transportation
  systems with human-driven and automated vehicles,'' in \emph{Decision and
  Control (CDC), 2015 IEEE 54th Annual Conference on}.\hskip 1em plus 0.5em
  minus 0.4em\relax IEEE, 2015, pp. 4688--4694.

\bibitem{c11_liu2016towards}
X.~Liu and P.~Kumar, ``Towards safety of transportation systems with a mixture
  of automated and human-driven vehicles,'' in \emph{Communication Systems and
  Networks (COMSNETS), 2016 8th International Conference on}.\hskip 1em plus
  0.5em minus 0.4em\relax IEEE, 2016, pp. 1--3.

\bibitem{hsieh2017throughput}
P.-C. Hsieh, X.~Liu, J.~Jiao, I.~Hou, Y.~Zhang, P.~Kumar \emph{et~al.},
  ``Throughput-optimal scheduling for multi-hop networked transportation
  systems with switch-over delay,'' in \emph{Proceedings of the 18th ACM
  International Symposium on Mobile Ad Hoc Networking and Computing}.\hskip 1em
  plus 0.5em minus 0.4em\relax ACM, 2017, p.~16.

\bibitem{c20_dresner2007sharing}
K.~M. Dresner and P.~Stone, ``Sharing the road: Autonomous vehicles meet human
  drivers.'' in \emph{IJCAI}, vol.~7, 2007, pp. 1263--1268.

\bibitem{bento2013intelligent}
L.~C. Bento, R.~Parafita, S.~Santos, and U.~Nunes, ``Intelligent traffic
  management at intersections: Legacy mode for vehicles not equipped with v2v
  and v2i communications,'' in \emph{Intelligent Transportation Systems-(ITSC),
  2013 16th International IEEE Conference on}.\hskip 1em plus 0.5em minus
  0.4em\relax IEEE, 2013, pp. 726--731.

\bibitem{c19_qian2014priority}
X.~Qian, J.~Gregoire, F.~Moutarde, and A.~De~La~Fortelle, ``Priority-based
  coordination of autonomous and legacy vehicles at intersection,'' in
  \emph{Intelligent Transportation Systems (ITSC), 2014 IEEE 17th International
  Conference on}.\hskip 1em plus 0.5em minus 0.4em\relax IEEE, 2014, pp.
  1166--1171.

\bibitem{altche2017algorithm}
F.~Altch{\'e}, X.~Qian, and A.~de~La~Fortelle, ``An algorithm for supervised
  driving of cooperative semi-autonomous vehicles,'' \emph{IEEE Transactions on
  Intelligent Transportation Systems}, vol.~18, no.~12, pp. 3527--3539, 2017.

\bibitem{ahn2018safety}
H.~Ahn and D.~Del~Vecchio, ``Safety verification and control for collision
  avoidance at road intersections,'' \emph{IEEE Transactions on Automatic
  Control}, vol.~63, no.~3, pp. 630--642, 2018.

\bibitem{c6_kim2014mpc}
K.-D. Kim and P.~R. Kumar, ``An mpc-based approach to provable system-wide
  safety and liveness of autonomous ground traffic.'' \emph{IEEE Trans.
  Automat. Contr.}, vol.~59, no.~12, pp. 3341--3356, 2014.

\end{thebibliography}
\end{document}